\documentclass[11pt]{article}
\usepackage[margin=1in, centering]{geometry}

\usepackage{cite}
\usepackage{comment}
\usepackage{mathtools}
\usepackage{amssymb}
\usepackage{amsthm}
\usepackage{xcolor}
\usepackage{bbm}
\usepackage[parfill]{parskip}

\usepackage[colorlinks=true, citecolor=blue]{hyperref}
\usepackage[capitalise]{cleveref} 

\newtheorem{theorem}{Theorem}[section]
\newtheorem{prop}[theorem]{Proposition}
\newtheorem{lemma}[theorem]{Lemma}

\newtheorem{definition}[theorem]{Definition}

\newcommand{\F}{\mathbb{F}}  
  
\newcommand{\R}{\mathbb{R}}

\DeclareMathOperator*{\E}{\mathbb{E}}  
 
\newcommand{\1}{ \mathbbm{1}}
\newcommand{\one}[1]{\mathbf{1}_{\left\{#1\right\}}}

\newcommand{\eps}{\varepsilon}

\newcommand{\ip}[2]{\left\langle#1 ,#2 \right\rangle}

\newcommand{\bits}{\{0,1\}}

\def\BibTeX{{\rm B\kern-.05em{\sc i\kern-.025em b}\kern-.08em
    T\kern-.1667em\lower.7ex\hbox{E}\kern-.125emX}}
\begin{document}

\title{More efficient sifting for grid norms, and applications to multiparty communication complexity}

\date{}

 \author{
Zander Kelley\thanks{
Institute for Advanced Study (IAS).
Supported by the Director's Discretionary Fund. 
Email: \texttt{awk@ias.edu}.
} \and
Xin Lyu\thanks{
UC Berkeley. Supported by a Google Fellowship. Email: \texttt{xinlyu@berkeley.edu}
}
}

\maketitle

\begin{abstract}
Building on the techniques behind the recent progress on the 3-term arithmetic progression problem \cite{KelleyM2023strong}, Kelley, Lovett, and Meka \cite{KelleyLM2024-nof} constructed the first explicit 3-player function $f:[N]^3 \rightarrow \{0,1\}$ that demonstrates a strong separation between randomized and (non-)deterministic NOF communication complexity. Specifically, their hard function can be solved by a randomized protocol sending $O(1)$ bits, but requires $\Omega(\log^{1/3}(N))$ bits of communication with a deterministic (or non-deterministic) protocol.

We show a stronger $\Omega(\log^{1/2}(N))$ lower bound for their construction. To achieve this, the key technical advancement is an improvement to the sifting argument for grid norms of (somewhat dense) bipartite graphs. In addition to quantitative improvement, we qualitatively improve over \cite{KelleyLM2024-nof} by relaxing the hardness condition: while \cite{KelleyLM2024-nof} proved their lower bound for any function $f$ that satisfies a strong two-sided pseudorandom condition, we show that a weak one-sided condition suffices. This is achieved by a new structural result for cylinder intersections (or, in graph-theoretic language, the set of triangles induced from a tripartite graph), showing that any small cylinder intersection can be efficiently covered by a sum of simple ``slice'' functions. 
\end{abstract}

\section{Introduction}

In a recent work by Kelley and Meka \cite{KelleyM2023strong}, substantial quantitative progress was made on a well-known problem in additive combinatorics concerning the size of the largest set of integers $A \subseteq [N]$ which do not contain any three points $x,y,z$ which are evenly spaced (i.e., which form a $3$-AP). It was shown that such sets cannot have more than ``quasipolynomial" density: 
$|A| \leq 2^{-\Omega(\log(N)^{1/12})} N$.\footnote{See \cite{bloom2023kelley}, \cite{bloom2023improvement} for some subsequent improvements. }
This comes reasonably close to the size of known constructions of large $3$-AP free sets, which have size $|A| \geq 2^{-\Omega(\log(N)^{1/2})} N$.\footnote{
\cite{behrend46}. See also \cite{elkin2010improved}, \cite{gw10}, \cite{obryant11}, \cite{hunter2024new}, \cite{elsholtz2024improving} for some refinements.}

Since then, there have been several new applications in theoretical computer science and extremal combinatorics, which rely not on \textit{result} of this work but on its \textit{techniques}. 
Specifically, it has since been realized (first in a work of Kelley, Meka and Lovett \cite{KelleyLM2024-nof}) that the two main new ingredients from \cite{KelleyM2023strong} (``sifting" \cite[Section 4]{KelleyM2023strong}) and ``spectral positivity" \cite[Section 5]{KelleyM2023strong}) can be generalized to the graph theoretic setting, and are in fact best thought of as tools which are useful for the study of arbitrary bipartite graphs that have (at least) quasipolynomial edge-density. The main result of \cite{KelleyLM2024-nof} was an application of these new graph-theoretic tools to a problem in multiparty Number-On-Forehead (NOF) communication complexity. Subsequent additional applications (and new refinements) to these techniques were given in \cite{abboud2024new} and \cite{filmus2024sparse}. 

The specific applications in these works are on distant enough topics\footnote{For example, \cite{KelleyLM2024-nof} is about lower bounds for communication complexity, while \cite{abboud2024new} is about upper bounds for graph algorithms.} that it becomes somewhat difficult to give a succinct description of what these new techniques are \textit{good for}, precisely. Nevertheless, we offer the following attempt at one. In any situation where one is faced with a graph-theoretic problem which can be solved with Szemerédi's regularity lemma, on graphs with $\Omega(1)$ edge-density, then it is potentially reasonable to hope that this new machinery can lead to a solution for graphs with quasipolynomial edge-density. It is of course likely that for certain applications this hope is overly-optimistic, but we do not have a clear picture of what such cases should look like. The limitations of these new graph theoretic techniques are not yet well-understood.

This line of work (which develops graph-theoretic generalizations of the ideas introduced in \cite{KelleyM2023strong}) is still relatively new, and it is unclear that any of the known results or underlying technical lemmas represent the ``best version of themselves", in a quantitative sense. In particular, each work mentioned so far proves some result about structures within a finite space of size $N$, which holds for all structures of quasipolynomial density: at least $2^{-\log(N)^c}$ for some constant $c > 0$. But, we are not aware of any application where it can be demonstrated that the result is \textit{tight}, and we obtain the \textit{right} quasipolynomial. 
In fact, it would be interesting to describe any application at all, along with an efficient-enough implementation of these new graph-theoretic techniques, where it can be shown that we obtain a tight answer, with the right quasipolynomial bound.

A core technical component needed in each of the works mentioned is an efficient ``inverse theorem" for the grid norm (see \eqref{equ: defn grid norm} for its definition) of a bipartite graph: one needs to show (often by an argument involving sifting/dependent random choice) that for any dense bipartite graph with unusually large grid norm, it is possible to zoom in on some (reasonably sized) induced subgraph which has noticeably increased edge-density.
Our main technical contribution is to give an improvement regarding this step.\footnote{More specifically, our approach gives a more efficient inverse theorem for the $U(2,k)$ grid norm, which is the case most important for applications thus far. Our approach can be made to work also for general $U(\ell,k)$ grid norms, but for large $\ell$ (say, $\ell \approx k$) we do not seem to get any improvement over existing arguments.}  The specific result is given by \cref{lemma: grid upper bound for spread matrix}.\footnote{See also \cref{equ:deviation fg to uniform} for a related statement which is often more directly useful in applications. } 
As an application, we give a quantitative improvement to the main result of \cite{KelleyLM2024-nof}, which was the construction of an explicit function $\1_{D} : [N]^3 \rightarrow \bits$ which is hard for deterministic 3-NOF protocols but easy for randomized 3-NOF protocols. In particular, 
our improved analysis shows that, in order to express the hard set $D$ as a union of some $R$ cylinder intersections, $R \geq 2^{\Omega(\log(N)^{1/2})}$ is required. 
This improves the prior result, $R \geq 2^{\Omega(\log(N)^{1/3})}$, from \cite{KelleyLM2024-nof}.

\paragraph{Cylinder Intersections.}
The basic combinatorial object of interest in the study of multiparty NOF communication complexity is the \emph{cylinder intersection}. 
In the $3$-party setting, a cylinder intersection $S$ is a subset of a product space $X \times Y \times Z$ 
whose indicator function has the form
$$ \1_{S}(x,y,z) = f(x,y)g(x,z)h(y,z). $$
Here, $f,g,$ and $h$ are indicator functions, each of which depends on only $2$ out of $3$ coordinates (so, $S$ is the intersection of three ``combinatorial cylinders"). 
For the purpose of obtaining complexity lower bounds for NOF communication protocols, 
it is important to understand the structural limitations of cylinder intersections (or at least, some particular limitation, which we can try to target). 

Our main result is the following structural statement about small cylinder intersections (over $X \times Y \times Z$), which states that they can be somewhat efficiently covered by a small collection of simpler functions. To state it, we first describe this family of simple functions. 

\begin{definition}[Slice Function]
    A slice function $s(x,y,z)$  (on $\Omega := X \times Y \times Z$) is an indicator function of the form
    $$ s(x,y,z) = g(x,y)h(z) \quad \text{or} \quad g(x)h(y,z) \quad \text{or} \quad g(x,z)h(y).$$ 
\end{definition}

That is, a slice function is just a combinatorial rectangle on $(X \times Y) \times Z$ or $X \times (Y \times Z)$ or $(X \times Z) \times Y$.

\begin{theorem}[Slice Function ``Removal Lemma"]\label{thm: removal lemma}
Let $f(x,y),g(x,z),h(y,z)$ be some indicator functions on $\Omega := X \times Y \times Z$, and consider the cylinder intersection
$$ F(x,y,z) := f(x,y)g(x,z)h(y,z). $$
Suppose that $F$ is small, i.e.\
$$ \E_{(x,y,z) \in \Omega}[F(x,y,z)] \leq 2^{-d}. $$
Then, $F$ has a reasonably small (pointwise) cover $F'$,
$$ F(x,y,z) \leq F'(x,y,z) := \sum_{i=1}^{R} s_i(x,y,z), $$
for some slice functions $s_i(x,y,z)$. Specifically,
\begin{itemize}
    \item $\E_{(x,y,z) \in \Omega}[F'(x,y,z)] \leq 2^{-\Omega(d^{1/2})}\log|\Omega|$, and
    \item $\E[s_i]\ge 2^{-O(d)}$ for every $i\in [R]$ (consequently, $R \le 2^{O(d)}\cdot \log(|\Omega|)$). \footnote{It is possible to obtain a version of this statement with no dependence on the size of the ambient space $\Omega$ -- that is, if one is willing to settle for a \emph{fractional} cover by slice functions. See \cref{thm: fractional cover}.}
\end{itemize}
\end{theorem}

\paragraph{Explicit Separations for Deterministic and Randomized NOF Complexity.}
As a corollary, we obtain the following improvement to the main result from \cite{KelleyLM2024-nof}, which was the construction of an explicit function which is hard for deterministic 3-NOF protocols but easy for randomized 3-NOF protocols. 

\begin{theorem}\label{thm: main nof lb}
Let $q$ be a prime power and $k$ a large enough constant. Let $N=q^k$, and identify $[N] \cong \F_q^k$. Consider the following 3-player function over $[N] \times [N] \times [N]$. 
$$
D(x,y,z) = \1[\ip{x}{y}=\ip{x}{z}=\ip{y}{z}],
$$
Then:
\begin{enumerate}
    \item The randomized NOF communication complexity of $D$ (with public randomness) is $O(1)$.
    \item The deterministic (or non-deterministic) NOF communication complexity of $D$ is $\Omega((\log N)^{1/2})$.
\end{enumerate}
\end{theorem}

\paragraph*{A pseudorandomness perspective.} Quantitative improvements aside, we argue that the statement of \cref{thm: removal lemma} is interesting for some additional qualitative reasons. 
Firstly, it succinctly captures the exact weakness of cylinder intersections which was targeted in \cite{KelleyLM2024-nof} to obtain a communication lower bound.\footnote{Indeed, it can be checked that the key pseudorandomness statement, Theorem 2.4 from \cite{KelleyLM2024-nof}, follows directly from our slice function removal lemma.}
Additionally, it is interesting to compare the statement with some structurally related results in extremal combinatorics and unconditional pseudorandomness. 

For example, consider the collection of indicator functions $F$ defined on the space $\Omega = \bits^{n}$ which have the form
$$ F(x_1, x_2, \ldots, x_n) = \prod_{S \in \binom{[n]}{w}} f_{S}(x_S). $$
It was shown by Bazzi\footnote{\cite{bazzi2009polylogarithmic}. See also \cite{razborov2009simple}, \cite{de2010improved}, \cite{tal2017tight}.} that these functions (which are equivalently just the collection of width-$w$ CNFs) have upper and lower pointwise approximators
 $F_{\ell}(x) \leq F(x) \leq  F_{u}(x), $
which give a good approximation on average (for uniformly random points $x \in \{0,1\}^n$) in the sense that $\E_{x\in\bits^n} [F_u(x) - F_\ell(x)] \leq \eps$.
Furthermore, $F_\ell(x)$ and $F_u(x)$ are \textit{simple functions} -- specifically, polynomials of degree at most $k = \text{poly}(w,\log(n/\eps))$. 
He uses this structural result to prove that $k$-wise independent distributions are pseudorandom for width-$w$ CNFs, which follows because $k$-wise independent distributions are pseudorandom for degree $k$ polynomials. Our structural result can be applied in much the same way: it implies that any set $D$ which is evasive\footnote{By ``evasive", we just mean a set $D$ with a one-sided pseudorandomness property: no simple test function which captures a small fraction of $\Omega$ should be able to capture a large fraction of $D$.} for small slice functions must also be reasonably evasive for small cylinder intersections -- this fact is all that is needed for the proof of \cref{thm: main nof lb}.

It was noted in \cite{de2010improved} that the existence of good degree-$k$ sandwiching approximators $F_{\ell}(x) \leq F(x) \leq  F_{u}(x)$, for some function $F$ is, in fact, \emph{equivalent} to the fact that every $k$-wise independent distribution is pseudorandom for $F$ (with two-sided, additive error). The argument is based on linear programming duality, and is in fact quite generic: it extends easily to other settings involving a notion of pseudorandomness and a corresponding family of simple test functions. 

An interesting consequence of this is that, if we have some proof of a pseudorandomness statement which is ``abstract enough" (for example, it works for \emph{all} $k$-wise independent distributions, rather than just a particular construction), then we can actually extract a \emph{structural result} by linear programming duality. 
This is roughly how we obtain our main result, \cref{thm: removal lemma}. First, we argue that an arbitrary distribution which is evasive for slice functions must be reasonably evasive for cylinder intersections. Then, by a different linear programming duality argument (which is needed for the one-sided and multiplicative quantification of error relevant here), we extract a good fractional cover, and then round it to obtain a cover. 

\paragraph{Triangle Removal Lemma.}
Consider the following (unusual) formulation of Szemerédi's triangle removal lemma. For the purpose of comparing with \cref{thm: removal lemma}, we formulate it as a statement \textit{about cylinder intersections}. 

\begin{theorem}[Triangle Removal Lemma] \label{thm: triangle removal lemma}
Let $f(x,y),g(x,z),h(y,z)$ be some indicator functions on $\Omega := X \times Y \times Z$, and consider the cylinder intersection
$$ F(x,y,z) := f(x,y) g(x,z) h(y,z). $$
Suppose that $F$ is small, i.e.\
$$ \E_{(x,y,z) \in \Omega}[F(x,y,z)]  \leq \eps. $$
Then $F$ has a reasonably small (pointwise) cover $F'$, which has the form
$$ F(x,y,z) \leq F'(x,y,z) := f'(x,y) + g'(x,z) + h'(y,z). $$
Specifically,
$$ \E_{(x,y,z) \in \Omega}[F'(x,y,z)]  \leq \delta := \delta(\eps) $$
for some $\delta(\eps)  \rightarrow 0$ as $\eps \rightarrow 0$.
\end{theorem}

Recall that the standard formulation of the triangle removal lemma is stated in the language of graphs: it says that if an $N$-vertex graph has a small number of triangles (i.e., at most $\eps N^3$), then it is possible to remove a reasonably small number of edges (i.e., at most $\delta N^2$) in order to eliminate all triangles. One can observe that the two formulations differ only cosmetically.\footnote{Perhaps the only non-obvious step is to notice that the triangle removal lemma is equivalent (up to constant factors) to the special case of itself on tripartite graphs.}

It is not known what should be the ``right" quantitative dependence of $\delta$ on $\eps$ in the triangle removal lemma. Indeed, there is a large gap between the known upper bounds on $\delta(\eps)$ (which are tower-type, see e.g.\ \cite{fox2011new}), and the best known lower bounds on plausible values for $\delta(\eps)$, which follow from Behrend-type constructions of large sets of integers $A \subseteq [N]$ without $3$-term arithmetic progressions. In light of these constructions, the best plausible dependence of $\delta$ on $\eps$ would be (roughly) 
$$ \delta(\eps) \approx 2^{-\sqrt{\log(1/\eps)}}. $$
Coincidentally, this corresponds to the same parameters we obtain in our slice function removal lemma. So, to conclude our comparison, we note the following. At present, the statements (\cref{thm: removal lemma} and \cref{thm: triangle removal lemma}) are incomparable: the former obtains substantially better parameters, but the latter provides a cover by a (strictly) simpler collection of functions -- \emph{cylinders} instead of \emph{slice functions}. However, in the event that dramatically improved bounds for the triangle removal lemma were to be discovered, one could conceivably obtain a result strictly improving both statements.\footnote{We point out that, as far as we are aware, it could be possible to give a strict \emph{quantitative} improvement to just our slice function removal lemma itself. That is, we do not know of any sort of obstruction (Behrend-type or otherwise) to obtaining substantially improved parameters when we cover by slice functions. It would be interesting to find such an obstruction if there is one.}

\paragraph{An Open Question for $k$-party NOF.}
It is interesting to understand to what extent the techniques introduced in \cite{KelleyLM2024-nof} can be extended to obtain explicit separations for $k$-party deterministic and randomized NOF protocols, for $k \geq 4$. Concretely, we ask: can our covering result \cref{thm: removal lemma}, can be extended to apply to cylinder intersections $F(x_1, x_2, \ldots, x_k)$ defined over larger product spaces $X_1 \times X_2 \times \cdots \times X_k$? Can we also obtain a reasonably small cover $F'$ of $F$ by some slice functions? What about covering by ``generalized rectangles" -- functions of the form $g(x_{S})h(x_{T})$ for some partition $[k] = S \cup T$?
It seems that such questions are related to the extent to which the new graph-theoretic machinery (mentioned in the introduction) can be extended to \emph{hypergraphs}. If so, it seems likely this would be useful for applications beyond multiparty communication complexity, as well.

\section{Preliminaries}\label{sec: preliminary}

We usually work with a finite universe $\Omega$. We write $w\sim \Omega$ to denote that $w$ is drawn from $\Omega$ uniformly at random. Alternatively, we also use $\mathcal{U}$ to denote the uniform distribution. A function $p:\Omega\to \mathbb{R}_{\ge 0}$ is called a probability density function if $\E_{w\sim \Omega}[p(w)] = \sum_{w} \frac{p(w)}{|\Omega|} = 1$. To sample from $p$, we draw each $w\in \Omega$ with probability $\frac{p(w)}{|\Omega|}$. For any set $S\subseteq \Omega$, $w\sim S$ denotes the uniform distribution over $S$. Finally, we will frequently use the following change-of-distribution identity:
\begin{align}
 \E_{w\sim p}[S(w)] = \frac{|S|}{|\Omega|} \cdot \E_{w\sim S}[p(w)]. \label{equ:change-of-dist}
\end{align}
Let $X$ be a finite domain. The $\ell_p$-norm of a vector $u\in \mathbb{R}^X$ is defined as $\|u\|_{p} \coloneqq \E_{x\sim X}[|u_x|^p]^{1/p}$. For two vectors $u,v\in \mathbb{R}^X$ of the same dimension, their inner product is $\langle u, v\rangle \coloneqq \E_{x\sim X} [u_x \cdot v_x]$. For every $p\ge 1$, let $p^\star \coloneqq \frac{p}{p-1}$ be the H\"older conjugate of $p$. The H\"older inequality says that $\langle f, g\rangle \le \| f\|_p \cdot \|g \|_{p^\star}$. We collect some elementary equations regarding H\"older conjugates that we frequently use without notice:
\begin{align}
\frac{1}{p-1} = p^\star - 1 \text{~~~~,~~~~} p^\star = \frac{p}{p-1} \text{~~~~,~~~~}  \frac{1}{p} - 1 = -\frac{1}{p^\star}.
\end{align}
For a matrix $M\in \mathbb{R}^{X\times Y}$ and vectors $u\in \mathbb{R}^X, v\in \mathbb{R}^Y$, we define the \emph{normalized} matrix-vector product $Mg\in \mathbb{R}^{X}$ as $(Mg)(x) = \E_{y}[M(x,y)\cdot g(y)]$ and vector-matrix-vector product $f^\top M g = \E_{x,y}[f(x) M(x,y) g(y)]$. Finally, given $\ell, k\ge 1$, the $U(\ell, k)$-grid norm of $M$ is defined as
\begin{align}
    \|M\|_{U(\ell, k)} := \left( \E_{x_{1:\ell}} \E_{y_{1:k}} \left[ \prod_{i\in [\ell]}\prod_{j\in [k]} M(x_i, y_j) \right] \right)^{1/\ell k}. \label{equ: defn grid norm}
\end{align}

\subsection{Norms and Flat Norms of Vectors}

For a vector $v\in \mathbb{R}_{\ge 0}^{X}$, we may define its $\ell_p$-norm by
\begin{align}
\| v \|_{p} \coloneqq \E_{x\sim X} [(v_x)^p ]^{1/p} = \max_{u\in \mathbb{R}^X}\left[ \frac{\langle v, u \rangle}{\|u\|_{p^\star}} \right]. \label{equ: def vector norm}
\end{align}

We consider a related notion of the \emph{flat} $p$-norm\footnote{
This analytic quantity is closely related to the \emph{weak} $p$-norm, and is well known, but seems to not go by any common name. See e.g.\ Exercise 1.1.12 in \cite{grafakos2008classical}.
}, defined as
\begin{align}
\|v\|_{\overline{p}} \coloneqq \max_{S\subseteq X} \left\{ \| v|_S \|_1 \left( \frac{|S|}{|X|} \right)^{1/p} \right\}. \label{equ:def vector flat norm}
\end{align}
Thus, the flat $p$-norm is just the result one obtains by writing down the dual characterization of the $p$-norm, but allowing only \emph{flat} witnesses: vectors $u$ of the form $u = \1_{S}$ for some set $S$. 
Here, $v|_S\in \mathbb{R}^S$ denotes the restriction of $v$ onto the subset $S$ of coordinates (and $\| v|_S \|_1$ denotes the $1$-norm normalized relative to $S$, i.e.,  $\| v|_S \|_1 = \E_{x \sim S} |v_{x}|$).

The following proposition states the flat $p$-norm and the usual $\ell_p$-norm are closely related for all \emph{smooth} vectors.
\begin{prop}\label{prop:vector norm vs flat norm}
    Let $k\ge 1$ be a parameter. Suppose $v\in \mathbb{R}_{\ge 0}^{X}$ is such that $\frac{\|v\|_{\infty}}{\|v\|_{1}}\le 2^k$. Then, for every $\eps \in (0, 1/10)$ and $p\ge \frac{10 \log(k)}{\eps}$, it holds that
    \begin{align}
        \|v\|_{p} \ge \| v \|_{\overline{p}} \ge (1-\eps) \| v \|_p.
    \end{align}
\end{prop}

In fact, the first inequality holds for any $p\ge 1$ without any restriction on $p$.

\begin{proof}
    To see the first inequality, simply note that for any $S\subseteq X$, we may consider the indicator vector $f\in \mathbb{R}^{X}$ defined by $f_x = \one{x\in S}$. We have
    \begin{align}
        \frac{\langle v, f\rangle}{\|f\|_{p^\star}} = \frac{|S|}{|X|} \cdot \| v|_S \|_1 \cdot \left( \frac{|S|}{|X|} \right)^{-\frac{1}{p^\star}}  = \| v|_S\|_1 \left( \frac{|S|}{|X|} \right)^{\frac{1}{p}}.  \notag
    \end{align}
    This means any candidate maximizer $S\subseteq X$ for $\|v\|_{\overline{p}}$ gives a candidate maximizer for $\|v\|_{p}$. Optimizing over all possible $S$ proves the inequality.

    We turn to establish the second inequality. By properly scaling, we can assume W.L.O.G. $\|v\|_{\infty} = 1$. With this in mind, we consider a geometric decomposition 
    \begin{align}
        v = \sum_{j=0}^{-\infty} v^j \notag
    \end{align}
    where $v^j$ is the vector that contains entries of $v$ with value $(2^{\eps (-j-1)},2^{\eps j}]$. Namely, $(v^j)_x = v_x$ if $v_x \in (2^{\eps (-j-1)},2^{\eps j}]$ and $v^j_x = 0$ otherwise. Since $v^j$'s have disjoint support, we have
    \begin{align}
        \|v\|_p^p = \sum_{j=0}^{\infty} \|v^j\|_p^p. \notag
    \end{align}
    For a suitable $d \ge \frac{\log(\frac{4\|v\|_{\infty}}{\|v\|_p})}{\eps} \approx \frac{k}{\eps}$, we have
    \begin{align}
        \left\| \sum_{j \ge d} v^j \right\|_{p} \le \left\| \sum_{j \ge d} v^j \right\|_{\infty} \le 2^{-\eps d} \le \frac{\|v\|_p}{4}. \notag
    \end{align}
    Consequently,
    \begin{align}
        \sum_{j=0}^{d-1} \left\|  v^j \right\|_{p}^p 
        &= \| v\|_p^p - \left\| \sum_{j\ge d} v^j \right\|_{p}^p \ge \frac{1}{2} \|v\|_p^p. \notag
    \end{align}
    By averaging principle, there exists $j^*<d$ such that $\left\|  v^{j^*} \right\|_{p}^p \ge \frac{1}{2d} \|v\|_p^p$. Letting $S$ be the support of $v^{j^*}$, we obtain
    \begin{align}
        &~~~~\| v|_S \|_1 \cdot \left( \frac{|S|}{|X|} \right)^{\frac{1}{p}} \ge 2^{(-j^*-1)\eps } \left( \frac{|S|}{|X|} \right)^{\frac{1}{p}} \notag \\
        &\ge  2^{\eps } \| v^{j^*} \|_p  \ge 2^{\eps }\cdot (\frac{1}{2d})^{1/p}\cdot \|v\|_p \notag \\
        &\ge (1-O(\eps)) \| v\|_p. \label{equ: flat vs norm conclusion}
    \end{align}
    To replace $(1-O(\eps))$ in \eqref{equ: flat vs norm conclusion} with a $(1-\eps)$, we can work with a smaller $\eps' = \Theta(\eps)$ in this argument. the proof will go through so long as we have, say, $p\ge \frac{2 d}{\eps'}$.
\end{proof}

\subsection{Norms and Flat Norms of Matrices}

We can similarly develop a norm-vs-flat norm connection for matrices and operator norms. In particular, consider a matrix $M\in \mathbb{R}^{X\times Y}$ whose rows and columns are indexed by $X$ and $Y$, respectively. We let $\|M\|_1 = \E_{x\sim X,y\sim Y}[M(x,y)]$ be the $1$-norm of $M$.

Given $\ell, r \ge 1$ and a \emph{non-negative} matrix $M$, define the \emph{flat} $(\ell, r)$-norm of $M$ as
\begin{align}
\| M \|_{\overline{\ell},\overline{r}} = \max_{S\subseteq X, T\subseteq Y} \left\{ \| M|_{S, T} \|_1 \cdot \left( \frac{|S|}{|X|} \right)^{\frac{1}{\ell}} \left( \frac{|T|}{|Y|}\right)^{\frac{1}{r}} \right\}. \label{equ: def matrix flat norm}
\end{align}

We can then compare the flat norms with the usual operator norm. In particular, we define the $(\ell,r)$-norm of a non-negative matrix $M$ as
\begin{align}
\| M \|_{\ell, r} = \max_{\substack{f:X\to \mathbb{R}_{\ge 0},\\ g:Y\to \mathbb{R}_{\ge 0}}} \left\{ \frac{f^\top M g}{ \|f\|_{\ell^\star} \|g\|_{r^\star}} \right\} \coloneqq \max_{\|g\|_{r^\star} \leq 1} \|M g \|_{\ell} . \label{equ: def matrix flat norm}
\end{align}
From an operator norm perspective, one might find it easier to understand $\|M \|_{\ell, r}$ as $\|M\|_{\ell\gets r^\star}$ (i.e., how $M$ maps a $r^\star$-norm bounded vector to a vector with bounded $\ell$-norm). 


The following lemma connects flat norms with operator norms for smooth matrices.

\begin{lemma}\label{lemma:flat-to-operator-norm}
Suppose $M\in {\R_{\ge 0}}^{X\times Y}$ is such that $\|M\|_{\infty} \le \|M\|_1 \cdot 2^d$. Then, for any $\eps\in (0, 1)$, provided that $\min(\ell, r)\ge \frac{20 \log d}{\eps}$, we have
\begin{align}
\|M\|_{\ell, r} \ge \|M\|_{\overline{\ell},\overline{r}} \ge (1-\eps)\cdot \| M \|_{\ell, r}. \label{equ:matrix flat vs operator norm} 
\end{align}
\end{lemma}

\begin{proof}
    To see the first inequality, simply note that for any $S\subseteq X,T\subseteq Y$, we may consider the indicator functions $f\coloneqq \mathbf{1}_{x\in S}$ and $g\coloneqq \mathbf{1}_{y\in T}$, and derive that
    \begin{align}
        &~~~~\frac{f^\top M g}{\|f\|_{\ell^\star}\|g\|_{r^\star}} \notag \\
        &= \| M\mid_{S\times T}\|_1\cdot \left(\frac{|S|}{|X|}\right)^{1-\frac{1}{\ell^\star}} \left( \frac{|T|}{|Y|} \right)^{1-\frac{1}{r^\star}} \notag \\
        &=  \| M\mid_{S\times T}\|_1\cdot \left(\frac{|S|}{|X|}\right)^{\frac{1}{\ell}} \left( \frac{|T|}{|Y|} \right)^{\frac{1}{r}}  \notag
    \end{align}
    This means any candidate subsets $S,T$ for $\|M\|_{\overline{\ell},\overline{r}}$ give a candidate maximizer for $\|M\|_{\ell, r}$ as well. Optimizing over all possible $S,T$ proves the first inequality.

    We now prove the second inequality. Assume without loss of generality that $\|M\|_1 = 1$. So we have $\|M\|_{\infty}\le 2^d$ and consequently $\|M\|_{\infty\gets p}\le 2^d$ for any $p\ge 1$. Say $\tau \coloneqq \|M\|_{\ell, r}$ is achieved by witnesses $f:X\to \mathbb{R}_{\ge 0}$ and $g:Y\to \mathbb{R}_{\ge 0}$ with $\|f\|_{\ell^\star} = \|g\|_{r^\star} = 1$. Consider the following two optimization problems
    \begin{align}
        \max_{h:X\to \mathbb{R}_{\ge 0}}\left\{ \frac{h^\top M g}{\|h\|_{\ell^\star}} \right\} \text{~~~~vs.~~~~} \max_{h:X\to \{0,1\}}\left\{ \frac{h^\top M g}{\|h\|_{\ell^\star} } \right\}\notag 
    \end{align}
    The optimum of the left problem is realized by taking $h=f$ and, by Proposition~\ref{prop:vector norm vs flat norm}, there exists $\overline{f}:X\to \{0, 1\}$ such that
    \begin{align}
        \frac{\overline{f}^\top M g}{\|\overline{f}\|_{\ell^\star} } \ge (1-\eps)  \frac{f^\top M g}{\|f\|_{\ell^\star} }, \notag
    \end{align}
    provided that $\ell \ge \frac{10 \log d}{\eps}$. Next, one can similarly find $\overline{g}:Y\to \{0,1\}$ (provided that $r\ge \frac{10 \log d}{\eps}$) such that
    \begin{align}
        \frac{\overline{f}^\top M \overline{g}}{\|\overline{f}\|_{\ell^\star} \| \overline{g}\|_{r^\star} } \ge (1-\eps) \frac{\overline{f}^\top M g}{\|\overline{f}\|_{\ell^\star} \| g\|_{r^\star} } \ge \frac{(1-O(\eps))\cdot \tau}{\|\overline{g}\|_{r^\star} \|\overline{f}\|_{\ell^\star}} . \notag
    \end{align}
    The sets indicated by $\overline{f},\overline{g}$ would then witness that $\|M\|_{\overline{\ell},\overline{r}} \ge (1-O(\eps)) \|M\|_{\ell, r}$, as desired.
\end{proof}

\section{Spread Matrices and Their Products}\label{sec: spread matrix}

This section refines, simplifies, and improves the technique from \cite[Section 4]{KelleyLM2024-nof}.

\subsection{Proof Overview}\label{sec: spread matrix overview}

Before developing the results, let us comment on how we obtain a better sifting argument than prior works \cite{KelleyM2023strong,KelleyLM2024-nof}. 

\paragraph*{Review of \cite{KelleyLM2024-nof}.} At the core of all the ``sifting'' arguments is a ``density increment'' argument. To review the \cite{KelleyLM2024-nof} proof: we start with a non-negative matrix $M\in \mathbb{R}^{N\times N}$ that is (1) bounded: $\|M\|_{\infty} = 1$ and (2) relatively dense: $\|M\|_1 = \E_{(x,y)}[M(x,y)] \ge 2^{-d}$. We inspect how ``pseudorandom'' $M$ is, as measured by its $U(2,k)$-grid norm $\|M\|_{U(2,k)}$. There are two cases: If $M$ is ``pseudorandom'' enough in the sense that $\| M \|_{U(2,k)} \le (1+\varepsilon) \|M\|_{1}$, we can proceed to argue other nice properties of $M$. Otherwise, $M$ is not pseudorandom and we have $\|M\|_{U(2,k)}\ge (1+\varepsilon)\|M\|_1$. In this case, one tries to find a sub-matrix $X'\times Y'\subseteq [N]\times [N]$, such that the average of $M$ increases by a multiplicative factor of $(1+\varepsilon)$ where $\eps$ is a small constant, and the sub-matrix is of size at least $\frac{1}{2^{O(k d)}}$. On the sub-matrix, one can iterate this argument. Since each iteration increases the density of $M$ by at least $(1+\varepsilon)$. The density increment step only happens for at most $d/\varepsilon$ times, and we are guaranteed a final submatrix of density at least $2^{-O(d^2 k)}$. In the most relevant parameter regime, one first fixes $d = \log(\frac{1}{\text{density of the object}})$ and chooses $k\approx d$. This leads to a $2^{O(d^3)}$ loss in the size of the final sub-matrix, and is ultimately the reason that \cite{KelleyLM2024-nof} ended up with a $\Omega((\log N)^{1/3})$ lower bound for the $3$-party NOF communication model.

\paragraph*{Our improvement.} We show that, with a new and simpler way to sift, one gets a trade-off between how much larger the mean of $M$ increases, and how much smaller the sub-matrix becomes. In particular, we show the following: in each iteration of the density increment argument, if the average of $M$ only increases by a single $(1+\eps)$ factor, then the sub-matrix $X'\times Y'$ only gets smaller by a factor of $2^{-O(k)}$. On the other hand, if we can only increase onto a sub-matrix $X'\times Y'$ of size much smaller than $2^{-\Omega(k)}$, then the average of $M$ will increase by a significantly larger amount. Roughly speaking, there always exists a parameter $1\le \gamma \le d$ and a sub-matrix $X'\times Y'$ of size $2^{-\gamma\cdot k}$ such that $\E_{(x,y)\sim X'\times Y'}[M(x,y)]$ increases by a factor of $(1+\eps)^{(d-\gamma )k}$. Iterating based on this claim allows us to find a sub-matrix $M$ with the desired pseudorandom property, yet the size of the final sub-matrix is at least $2^{-O(kd)}N^2$.

One might find it easier to understand the improvement by drawing an analogy to vector $\ell_k$-norms. Say we have a vector $v\in [0,1]^{N}$ that we know $\|v\|_1 \ge 2^{-d}$ and $\|v\|_k := \E[v_i^k]^{1/k} > (1+\eps) \E[v_i]$. This can happen for (at least) two possible reasons: (1) there is a subset $S\subseteq [N]$ such that $|S| \ge 2^{-k} N$ and $v_i > (1+\frac{1}{2})\|v\|_1$ for $i\in S$, or (2) there is a much smaller subset $S\subseteq [N]$ of size only $|S| > 2^{-dk} N$, but $v_i > 2^{d} \cdot \|v\|_1$ for $i\in S$. Both cases lead to a significantly larger $\ell_k$-norm of $v$. Still, to run an efficient ``density increment'' argument on $v$, one must account for the trade-off between $|S|$ and $\min_{i\in S}\{v_i\}$. 

Currently, our improvement only applies to the $U(2,k)$-grid norms\footnote{We can extend this to $U(\ell, k)$-norms by a tensoring trick. Still, the improvement is only significant in the regime that $\ell \ll k$.}, mostly due to the nice interpretation of $U(2,k)$ norm as ``the $k$-th moment of the pairwise inner products between rows of $M$'', as well as the flat-vs-operator norm connection established in Section~\ref{sec: preliminary}. It remains to be seen whether a similar improvement can be made for the case of 3-APs \cite{KelleyM2023strong}. 

\paragraph*{Organization of the section.} With the setup from Section~\ref{sec: preliminary}, the density increment argument is compactly presented in Section~\ref{sec: spread matrix - sifting}. We also refine the ``spectral positivity'' part of the \cite{KelleyLM2024-nof} argument and present a simplified proof in Section~\ref{sec: spread matrix - spectral positive}. Finally, Theorem~\ref{thm: product of spread matrices} combines the two ingredients and concludes this section.

\subsection{Spreadness of Matrices}\label{sec: spread matrix - sifting}

We define a notion of spreadness for matrices. Roughly, it says that a spread matrix has no large sub-matrix with significantly larger entries on average.

\begin{definition}\label{def:spread matrix}
    A matrix $M\in \mathbb{R}_{\geq 0}^{X\times Y}$ is $(t,\varepsilon)$-spread, if $\|M\|_{\overline{t},\overline{t}} \le \|M\|_1 \cdot (1+\eps)$.
\end{definition}

We compare Definition~\ref{def:spread matrix} with the definition of spreadness from \cite{KelleyLM2024-nof}. The spreadness of \cite{KelleyLM2024-nof} is parameterized by two parameters $k \mathbb{N}$ and $\eps > 0$. A matrix $M$ is spread in the sense of \cite{KelleyLM2024-nof}, if for any submatrix $M'\in \mathbb{R}^{X'\times Y'}$ of size $|X'\times Y'|\ge 2^{-k} | X\times Y|$, one has that $\|M'\|_1 \le \|M\|_1 \cdot (1+\varepsilon)$. Namely, the density of every ``reasonably large'' subractangle can only be larger than the density of $M$ be at most a $(1+\varepsilon)$ factor. Our Definition~\ref{def:spread matrix}, on the other hand, should be (roughly) read as for every submatrix $M'$ of size $(1-\varepsilon)^{k\cdot \tau}$ for some $\tau \ge 1$, one has $\|M'\|_1 \le \|M\|_1\cdot (1+\varepsilon)^{\tau}$. So, our definition is stronger in that we will consider submatrices of size significantly smaller than $2^{-k}$, and we impose proportionally more relaxed density upper bounds for smaller submatrices.

The key technical lemma of \cite{KelleyLM2024-nof} asserts that a spread matrix has bounded $U(2,k)$-norm. We state and prove our version of the lemma in the following.

\begin{lemma}\label{lemma: grid upper bound for spread matrix}
    Let $k,d\ge 1, \eps\in (0,1/5)$ be parameters such that $k\ge 20 d/\eps$. Suppose $M\in [0,1]^{X\times Y}$ is a matrix that is $(k,\eps)$-spread and $\|M\|_1 \ge 2^{-d}$. Then, $\| M \|_{U(2,k)} \le (1+O(\eps)) \| M \|_1$.
\end{lemma}

\begin{proof}
The condition of $M$ implies that $\frac{\|M\|_\infty}{\|M\|_1}\le 2^d$. As such, Definition~\ref{def:spread matrix} and Lemma~\ref{lemma:flat-to-operator-norm} imply that $\|M\|_{k,k} \le (1+O(\eps)) \|M\|_{\overline{k},\overline{k}} \le (1+O(\eps)) \|M\|_{1}$. For each $x\in X$, let $M_x$ be the $x$-th row vector of $M$. For a proper choice of $p\approx \frac{10d}{\eps}$, we have
\begin{align}
&~~~~ \| M \|_{U(2,k)}^{2k}  \\
&= \E_{x,x'\sim X}[ \langle M_x,M_{x'} \rangle^{k}] \notag \\
&= \E_{x\sim X} [ \| M \cdot M_x^\top  \|_k^k ] \notag \\
&\le \E_{x\sim X} \left[ \left(  \| M \|_{k, p} \cdot \|M_x\|_{p^\star}]  \right)^{k} \right]  \notag \\
&\le \| M \|_{k,k}^k \cdot \E_{x} \left[ \left( \| M_x \|_1^{1/p^{\star}} \| M_x \|_{\infty}^{\frac{p^\star -1}{p^\star}} \right)^{k}  \right] \notag \\
&\le\| M \|_{k,k}^k \cdot \| M \|_{\frac{k}{p^\star}, 1}^{\frac{k}{p^\star}} \cdot  2^{\frac{dk}{p}}  \notag \\
&\le \| M \|_{k,k}^{2k} \|M \|_{k,k}^{\frac{k}{p^\star} - k} \cdot 2^{\frac{dk}{p}}. \notag
\end{align}
Here, the fourth-line is using the definition of operator norm, and the last line is due to $\|M\|_{k,k}\ge \|M\|_{\frac{k}{p^*},1}$. Raising both sides to the power of $\frac{1}{2k}$ and using the fact $\|M\|_{k,k}\ge \|M\|_{1}\ge 2^{-d}$, we obtain
\begin{align}
    \| M \|_{U(2,k)} &\le \|M\|_{k,k}\cdot \|M\|_{k,k}^{-\frac{1}{p}}  \cdot 2^{\frac{d}{2p}} \notag \\
    &\le \|M\|_{k,k}\cdot 2^{2d/p} \notag\\
    &\le (1+O(\eps)) \|M\|_1.
\end{align}

This completes the proof.
\end{proof}

\subsection{Product of Spread Matrices is Near-Uniform}\label{sec: spread matrix - spectral positive}

For a matrix $f:X\times Y\to \mathbb{R}_{\ge 0}$, define $R_f:X\times Y\to \mathbb{R}_{\ge 0}$ by
$$
(R_f)(x,y) = \E_{y'\sim Y}[f(x,y')].
$$
Namely, $R_f$ is the row-averaging of $f$. We make the following claim.

\begin{lemma}\label{lemma:matrix-shift}
    Let $f:X\times Y\to \mathbb{R}_{\ge 0}$ be a matrix with $\| f - R_f\|_{U(2,k)}\ge 6\sqrt{\eps}\cdot \mathbb{E}[f]$. Suppose for every $x\in X$, it holds $\mathbb{E}_{y}[f(x,y)] \ge (1-\varepsilon) \mathbb{E}[f]$. Then, for $p \ge \frac{2k}{\varepsilon}$, we have $\|f\|_{U(2,p)}\ge (1+\varepsilon)\|f\|_1$.
\end{lemma}

We mostly use Lemma~\ref{lemma:matrix-shift} in its contrapositive form: Suppose we can prove a matrix $f$ is such that $\|f\|_{U(2,p)}\le (1+\eps) \|f\|_1$. Then, provided with the additional assumption $\E_y[f(x,y)] \ge (1-\eps)\E[f]$, we may conclude $\|f - R_f\|_{U(2,k)} \le O(\sqrt{\varepsilon} \|f\|_1)$.

To prove Lemma~\ref{lemma:matrix-shift}, we will need the following lemma from \cite{KelleyM2023strong}.

\begin{lemma}[See Proposition D.1 of \cite{KelleyM2023strong}]\label{lemma:spectral-positive}
    Let $\eps \in (0,1/4)$. Suppose $A$ is a random variable with $\mathbb{E}[A^k] \ge (2\varepsilon)^k$ for some even $k$ and $\mathbb{E}[A^t] \ge 0$ for every $t\ge 1$. Then, for all $p \ge \frac{k}{\eps}$, it holds that $\E[|1+A|^p] \ge (1+\eps)^{p}$. 
\end{lemma}

\begin{proof}[Proof of Lemma~\ref{lemma:matrix-shift}]
    Consider a random variable $A$ which takes value $\langle (f-R_f)(i,\star), (f-R_f)(j,\star)\rangle$ with probability $\frac{1}{|X|^2}$ for every $(i,j)$. We know that
    $$
    \E[A^k] = \|f-R_f\|_{U(2,k)}^{2k} \ge (6\sqrt{\eps} \E[f])^{2k},
    $$
    and that for every $t\ge 1$,
    $$
    \begin{aligned}
    \E[A^t] 
    &= \E_{y_1,\dots, y_t} \E_{x_1,x_2} \left[\prod_{x_i,y_j}(f-R_f)(x_i,y_j)]\right] \\
    &= \E_{y_1,\dots, y_t}  \left( \E_x \left[\prod_{y_j}(f-R_f)(x,y_j) \right] \right)^2 \\
    &\ge 0.
    \end{aligned}
    $$
    Let $\eps' = 2\eps - \eps^2$. Using Lemma~\ref{lemma:spectral-positive} on $\frac{A}{\E[f]^2(1-\eps')}$, we obtain that $\E[(A + (1-\eps')\E[f]^2)^p] \ge ((1+\eps)\E[f]^2)^p$ for every $p\ge k/\eps$. Observe that each row of $(f-R_f)$ sums to zero. Let $J = 1^{X\times Y}$ be the all-one matrix. We calculate
    \begin{align}
    &~~~~ \|f-R_f + (1-\eps)\E[f] J\|_{U(2,p)}^{2p} \notag \\
    &= \E_{i,j} \Big{\langle} (f-R_f + (1-\eps)\E[f] J)(i,\star),  \\
    &~~~~~~~~~  (f-R_f + (1-\eps)\E[f] J)(j,\star) \Big{\rangle}^p \notag \\
    &= \E_A[((1-\eps')\E[f]^2+A)^p] \notag \\
    & \ge ((1+\eps)\E[f]^2)^{p}. \notag
    \end{align}
    Recall that we have the entry-wise bound of $R_f(x,y) \ge 1-\eps$. Hence, for every $i,j$, we have
    \begin{align*}
    &~~~~~\langle f(i,\star), f(j,\star) \rangle \\
    &\ge \Big\langle (f-R_f + (1-\eps)\E[f] J)(i,\star), \\
    &~~~~~~~~~~~~~(f-R_f + (1-\eps)\E[f] J)(j,\star) \Big\rangle.
    \end{align*}
    This is to say, the random variable $A+(1-\eps')\E[f]$ is \emph{stochastically dominated} by $\langle f(\mathbf{i},\star),f(\mathbf{j}, \star)\rangle$.

    One may now be tempted to conclude that the $p$-th moment of $A+(1-\eps')\E[f]$ gives a lower bound for the $p$-th moment of $\langle f(\mathbf{i},\star),f(\mathbf{j}, \star)\rangle$. However, this is generally not true as the random variables can be negative (for example, the random variable $\mathbf{0}$ stochastically dominates $\mathbf{-1}$, but the latter has a larger $p$-th moment). Nonetheless, since $A$ satisfies that $\E[A^t]\ge 0$ for every $t\ge 1$, we know $\E[(A+(1-\eps') \E[f]^2)^p]\ge 0$ for every odd $p$. Therefore, as long as the $p$-th moment of $A+(1-\eps')\E[f]^2$ is concerned, the contribution from the ``positive'' instantiation of the r.v.~is no less than the contribution from the negative part. By only considering the non-negative instantiation of $A+(1-\eps')\E[f]^2$ and using stochastic domination, we may consequently conclude
    \begin{align*}
    \|f\|_{U(2,p)}^{2p} 
    &\ge \E_A[(A+(1-\eps')\E[f]^2)^p \cdot \mathbf{1}\{A\ge 1\}]\\
    &\ge \frac{1}{2} \E_{A}[(A+(1-\eps')\E[f]^2)^p] \\
    &\ge \frac{1}{2}((1+\eps)\E[f]^2)^{p} \\
    &\ge (1+ \eps - O(1/p))^p \E[f]^{2p}.
    \end{align*}
    as desired. For even $p$, we can work with $p-1$ and appeal to the monotonicity of norms.
\end{proof}

\paragraph*{Decoupling inequality.} The rest of the argument is largely similar to \cite{KelleyLM2024-nof}. We nonetheless give a self-contained exposition here. First, let $f:X\times Y\to \mathbb{R}_{\ge 0}, g:Z\times Y\to \mathbb{R}_{\ge 0}$ be two non-negative matrices that are both spread and smooth. We consider their product $f\circ g:X\times Z\to \mathbb{R}_{\ge 0}$ defined as
\begin{align}
    (f\circ g)(x,z) = \E_{y\sim Y}[f(x,y) \cdot g(z, y)]. \label{equ:def-matrix-product}
\end{align}
Assuming the spreadness and smoothness of $f,g$, we claim that $f\circ g$ is close to the matrix $R_f\circ R_g$ (which we note is a rank-$1$ matrix).

To prove the claim, we use a moment method. For any even $k\ge 1$ of our choice, we can write
\begin{align}
    &~~~~ \| f\circ g - R_f\circ R_g \|_k^k \notag \\
    &= \| (f-R_f) \circ (g-R_g) \|_k^k \notag \\
    &= \E_{x\sim X,z\sim Z} \left\langle (f-R_f)(x ,*), (g - R_g)(z, *)\right\rangle^k \notag \\
    &= \E_{x,z} \E_{y_1,\dots, y_k} \left( \prod_{1\le j\le k} (f-R_f)(x,y_j) (g-R_g)(z,y_j) \right) \notag \\
    &= \E_{y_1,\dots, y_k\sim Y} \left( \E_{x\sim X} \prod_{1\le j\le k} (f-R_f)(x,y_j) \right) \times \notag \\
    &~~~~~~~~~~~~~~~~~ \left( \E_{z\sim Z} \prod_{1\le j\le k} (g-R_g)(z,y_j) \right) \notag \\
    &\le \left(\E_{y_1,\dots, y_k\sim Y} \left( \E_{x\sim X} \prod_{1\le j\le k} (f-R_f)(x,y_j) \right)^2 \right)^{1/2}\times \notag \\
    &~~~~~~ \left(\E_{y_1,\dots, y_k\sim Y} \left( \E_{z\sim Z} \prod_{1\le j\le k} (g-R_g)(z,y_j) \right)^2 \right)^{1/2} \notag \\
    &\le \| f - R_f \|_{U(2,k)}^{k} \| g - R_g \|_{U(2,k)}^k. \label{equ:deviation fg to uniform}
\end{align}
Therefore, as long as we can establish that $f, g$ satisfy the prerequisites of Lemma~\ref{lemma:matrix-shift}, we can then upper bound the deviation $\| f\circ g - R_f \circ R_g \|_k$ by $O(\eps^2 \|f\|_1\|g\|_1)$, which would then imply that all but a tiny fraction (roughly $2^{-k}$) of pairs $(x, z)$ satisfy that $\left| \langle f(x,*),g(z,*)\rangle - R_f(x)\cdot R_g(z)\right| \le O(\eps^2 \|f\|_1\|g\|_1)$. 

Overall, we compile everything we have discussed so far and formulate the following theorem.
\begin{theorem}\label{thm: product of spread matrices}
Let $f\in [0,1]^{X\times Y}$ and $g\in [0,1]^{Z\times Y}$ be two matrices. Let $d\ge 10$ be an integer. Assume that
\begin{itemize}
    \item $f, g$ are $(100 d/\eps, \eps)$-spread,
    \item $\E_y[f(x,y)] \ge 1-\eps$, $\E_{y}[g(z,y)] \ge 1-\eps$ for every $x, z$, and
    \item $\|f\|_1,\|g\|_1\ge 2^{-d}$.
\end{itemize}
Then, it follows that
\begin{align*}
    \| f\circ g - R_f\circ R_g \|_{d} \le O(\eps\|f\|_1 \cdot \|g\|_1).
\end{align*}
\end{theorem}

\begin{proof}
    From the spreadness and density lower bound of $f$ we obtain that $\|f\|_{U(2,10d/\eps)}\le (1+O(\eps))\|f\|_1$ by Lemma~\ref{lemma: grid upper bound for spread matrix}. Using Lemma~\ref{lemma:matrix-shift} in its contrapositive form we then get $\|f-R_f\|_{U(2,d)}\le O(\sqrt{\eps} \|f\|_1)$. Repeat the same argument on $g$. We obtain $\|g-R_g\|_{U(2,d)}\le O(\sqrt{\eps} \|g\|_1)$. Finally, by the decoupling inequality as shown above, we conclude that
    \begin{align*}
        \|f\circ g - R_f\circ R_g\|_d \le \|f-R_f\|_{U(2,d)} \cdot \|g - R_g\|_{U(2,d)} \le O(\eps \|f\|_1\|g\|_1),
    \end{align*}
    as claimed.
\end{proof}

\section{Proof of the Removal Lemma}

In this section, we prove Theorem~\ref{thm: removal lemma} by proving the following ``fractional cover'' version.

\begin{theorem} \label{thm: fractional cover}
    Let $f(x,y), g(x,z), h(y,z)$ be some indicator functions on $\Omega = X\times Y\times Z$. Consider their cylinder intersection $F(x,y,z) = f(x,y)\cdot g(x,z)\cdot h(y,z)$. Let $d \ge 5$ and suppose $F$ is such that
    $$
    \mathbb{E}_{(x,y,z)\sim \Omega}[F(x,y,z)] \le 2^{-d}.
    $$
    Then, $F$ is covered by a sum of slice functions:
    $$
    F(x,y,z) \le \sum_{i=1}^{R} c_i\cdot  s_i(x,y,z),
    $$
    where each $s_i$ is of size at least $2^{-O(d)}$ and
    $$
    \E_{(x,y,z)\sim \Omega}\left[\sum_{i=1} c_i s_i(x,y,z) \right] \le 2^{-\Omega(d^{1/2})}.
    $$
\end{theorem}


\paragraph*{Rounding.} We sketch how to round the fractional cover from Theorem~\ref{thm: fractional cover} to prove Theorem~\ref{thm: removal lemma}. Say we have $F\le \sum_{i=1}^R c_i s_i(x,y,z)$. Consider now a distribution over slice functions where we draw $s_i$ with probability $\frac{c_i}{\sum_{i'}c_{i'}}$. Then, for a given $(x,y,z)\in F$, drawing a random $s_i$ from the distribution can cover $(x,y,z)$ with probability at least $\frac{1}{\sum_{i'}c_{i'}}$. Therefore, drawing $O(\log(|\Omega|)\cdot \sum_{i'}c_{i'})$ slices ensures that one of the slices covers $(x,y,z)$ with probability $1 - \frac{1}{|\Omega|^2}$. So, to obtain an integral covering, we just independently draw $O(\log(|\Omega|)\cdot \sum_{i'}c_{i'})$ functions from the distribution. By a simple union bound, the resulting cover will cover every $(x,y,z)\in F$ with high probability. By Markov's inequality, we can ensure that the total size of the integral cover is larger than the fractional cover by at most a factor of $O(\log(|\Omega|))$. We conclude by noting that $|\sum_{i'} c_{i'}| \le 2^{O(d)}$ (this follows from the upper bound of $\E\left[\sum_{i=1} c_i s_i(x,y,z) \right]$ and the lower bound of $\E[s_i]$). Hence, the number of slices functions in the integral cover is at most $2^{O(d)}\cdot \log(|\Omega|)$.

\paragraph*{Proof of Theorem~\ref{thm: fractional cover}.} We prove the theorem by duality of linear programming. Let $\mathcal{S}$ be the collection of all large slice functions (i.e., slice functions of size $2^{-O(d)}$). We construct a $\Omega\times \mathcal{S}$ matrix $A\in \mathbb{R}^{\Omega \times \mathcal{S}}$. Here, rows and columns are naturally indexed by $(x,y,z)\in \Omega$ and $s\in \mathcal{S}$. The entry indexed by row $(x,y,z)$ and column $s$ is equal to $s(x,y, z)$. We also define the weight of a slice $s\in \mathcal{S}$ to be $w(s) = \E[s]$. Slightly abusing notation, we understand $w$ as both a $\mathcal{S}$-dimensional column vector and a function mapping $\mathcal{S}$ to reals. Similarly, identify $F$ with a $\Omega$-dimensional column vector where the $(x,y,z)$-th entry of $F$ is $F(x,y,z)$. With this setup, let us consider the following linear program:
$$
\begin{aligned}
\min_{c\in (\mathbb{R}_{\ge 0})^{\mathcal{S}}} ~~~& \sum_{s} c_s \cdot w(s) \\
s.t. ~~~~& F \le A\cdot c.
\end{aligned}
$$
Intuitively, this is the covering problem considered in Theorem~\ref{thm: fractional cover}: we want to choose a weight $c_s$ for each slice function so that the weighted combination of slices covers every non-zero entry of $F$. Then, the conclusion of Theorem~\ref{thm: fractional cover} is equivalently saying that the optimal value of this problem is less than $2^{-\Omega(d^{1/2})}$.

We can now write down the dual problem:
$$
\begin{aligned}
\max_{p\in (\mathbb{R}_{\ge 0})^{\Omega}} ~~~~ &\sum_{(x,y,z)} F(x,y,z)\cdot p(x,y,z). \\
s.t. ~~~~ & w \ge A^\top \cdot p.
\end{aligned}
$$

Let us study this linear program more closely. This problem asks for a packing, subject to the constraint that every slice function $s$ upper bounds the weight placed inside the set indicated by $s$, and we would like to maximize the amount of weight placed in the set indicated by $F$. By strong duality, to prove Theorem~\ref{thm: fractional cover}, it suffices to show that every feasible solution $p$ to the dual problem has value at most $2^{-\Omega(d^{1/2})}$.

Suppose for the sake of contradiction that $p$ is a feasible solution with too large a value. We can assume W.L.O.G. that $p(x) = 0$ whenever $F(x) = 0$. Indeed, for any feasible $p$, it is easy to see that $(p\cdot F)$ is also feasible and attains the same value. To employ the pseudorandomness machinery, we would like interpret $p$ as a probability density function. 
Note that $p$ is non-negative by definition. After re-scaling $p$ by $\overline{p}=\frac{p}{\|p\|_1}$, we obtain a probability density function $\overline{p}$ and let $\mathcal{D}$ be the distribution described by $\overline{p}$.
Note that $\mathcal{D}$ is supported only on the support of $F$. Moreover, assuming $\sum_{x,y,z}p(x,y,z) \ge 2^{-c\cdot d^{1/2}}$ for some absolute $c > 0$, we have $\|p\|_1\ge |\Omega|^{-1}\cdot 2^{-c\cdot d^{1/2}}$. Thus, for any slice function $s$ it holds that
$$
\E_{\mathcal{D}}[s] \le\frac{|s|}{|\Omega|} \E_{s}[\overline{p}] \le \frac{1}{|\Omega|} \cdot \sum_{(x,y,z)\in s}p(x,y,z)\cdot |\Omega|\cdot 2^{c\sqrt{d}} \le 2^{c\sqrt{d}}\E_{\mathcal{U}}[s].
$$
In the following, we call any distribution $\mathcal{D}$ satisfying the condition above a $2^{c{d^{1/2}}}$-\emph{evasive} distribution to slice functions.

\begin{definition}\label{def: evasive distributions}
Let $k,d\in \mathbb{N}$. A distribution $\mathcal{D}$ supported on $X\times Y\times Z$ is $(k,d)$-evasive, if for every slice function $s$ of density $\|s\|_1 \ge 2^{-d}$, it holds that $\E_{w\sim \mathcal{D}}[s(w)] \le 2^{k} \E_{w\sim \mathcal{U}}[s(w)]$. A set $S\subseteq X\times Y\times Z$ is $(k,d)$-evasive, if the uniform distribution over $S$ is $(k,d)$-evasive.
\end{definition}

It has been clear from the above discussion that, to rule out the existence of a high-value dual solution, it suffices to prove the following proposition.

\begin{prop}\label{prop:evasive-to-largeness}
    There is an absolute constant $c > 0$ for which the following is true. Let $F$ be a cylinder interaction. Suppose $\mathcal{D}$ is a $(cd^{1/2}, d/c)$-evasive distribution supported on $F$.
    Then, it holds that $\E[F] > 2^{-d}\cdot |\Omega|$.
\end{prop}

\subsection{Proof of Proposition~\ref{prop:evasive-to-largeness}}

This subsection is devoted to the proof of Proposition~\ref{prop:evasive-to-largeness}.

In this subsection, we slightly abuse notation by using $p$ to denote the density function of $\mathcal{D}$. By Identity~\eqref{equ:change-of-dist} (namely, $\E_{w\sim \mathcal{D}}[s(w)] = \E_{w\sim s}[p(w)] \cdot{\E_{w\sim \mathcal{U}}[s(w)]}$), the condition of the proposition can be equivalently written as $\E_{w\sim s}[p(w)] \le 2^{c d^{1/2}}$ for every ``large'' slice function.

\paragraph*{Density increment step.} Let $t = c\cdot d^{1/2}$ where $c>0$ is a small but absolute constant. Also let $\eps > 0$ be a small constant. We would like to find a subcube $C^*$ such that the distribution $\mathcal{D}$ becomes spread in the subcube. An iterative greedy approach taken by \cite{KelleyLM2024-nof} is the following: initialize with $C^* = \Omega$, whenever these is a ``large'' subcube in which $\mathcal{D}$ has a noticably higher density, pass $C^*$ to the said subcube. Iterate this until no new ``density increment'' step can be made. 

It turns out that any greedy approach based on ``locally'' adjusting the solution can be conveniently carried out via a potential analysis: one defines a suitable potential function $\Phi$ measuring the ``quality'' of the solution. Then, one proves that the potential function strictly improves throughout the greedy procedure. When the procedure halts at a solution $C^*$, the solution achieves a local maximal for the potential. Moreover, once one has the potential function in hand, it is no longer necessary to explicitly track the greedy procedure. Instead, one can just \emph{define} $C^*$ to be a local optimum of the potential.

We will follow this plan. In particular, let $C^*$ be a subcube maximizing the following optimization problem:
\begin{align}
C^* = \arg\max_{C:|C|\ge 2^{-\frac{1}{c}t^2}|\Omega|} \left\{ \Phi(C) \coloneqq \E_{w\sim C}[p(w)] \cdot \left( \frac{|C|}{|\Omega|}\right)^{1/t} \right\}. \label{equ:cube-density-optimization}
\end{align}
Intuitively, we restrict to a smaller cube $C$ and try to balance between two terms: the ``penalty'' term of $\left( \frac{|C|}{|\Omega|}\right)^{1/t}$ and the ``reward'' term of $\E_{w\sim C}[p(w)]$.
We observe some properties regarding $C^*$. First, since $C=\Omega$ is a feasible solution to the problem, we have
\begin{align}
\Phi(C^*) \ge \Phi(\Omega) = \mathbb{E}_{\mathcal{U}}[p] = 1. \label{equ: C star phi lb}
\end{align}
Since every subcube $C$ with $|C|\ge 2^{-\frac{1}{c}t^2}|\Omega|$ is also a large slice function, we have $\E_{w\sim C}[p(w)]\le 2^t$. That is, the ``gain'' from $\E_{w\sim C^*}[p(w)]$ is upper bounded by $2^{t}$. Hence, it follows that $\left( \frac{|C^*|}{|\Omega|} \right)^{1/t}\ge 2^{-t}$ or equivalently $|C^*|\ge 2^{-t^2}|\Omega|$.

Now, write $C^* = X^* \times Y^* \times Z^*$. Let $p_{C^*}$ be the conditional distribution of $p$ onto $C^*$. By \eqref{equ: C star phi lb}, we have that $p_{C^*}(w) = \frac{p(w)}{\E_{w\sim C^*}[p(w)]} < p(w)$ for every $w\in C^*$. We define the projection of $p_{C^*}$ onto the three two-marginals (a.k.a.~faces):
$$
\begin{aligned}
f(x,y) &\coloneqq \E_{z}[p_{C^*}(x,y,z)], \\
g(x,z) &\coloneqq \E_{y}[p_{C^*}(x,y,z)], \\
h(y,z) &\coloneqq \E_{x}[p_{C^*}(x,y,z)].
\end{aligned}
$$
Note that by construction, we get $\|f\|_1 = \E_{(x,y)\sim X^*\times Y^*}[f(x,y)] = \E_{(x,y,z)\in C^*}[p_{C^*}(x,y,z)] = 1$ and similarly we have $\|g\|_1 = \|h\|_1 = 1$.

\paragraph*{The faces are spread.} We now show that each face is spread with respect to large rectangles. We take $f$ as a demonstrating example here and note that similar arguments apply equally well to $g$ and $h$. Our main claim is that, for any rectangle $X'\times Y'$ of density $\delta = \frac{|X'\times Y'|}{|X^*\times Y^*|}$, as long as $\delta > 2^{-t^2}$, it holds that
\begin{align}
\E_{(x,y)\sim X'\times Y'}[f(x,y)] \le \delta^{-1/t} \E_{(x,y)\sim X^*\times Y^*}[ f(x,y)] . \label{equ: f spread equ 1}
\end{align}
To briefly justify this, assume otherwise and observe that $(X'\times Y'\times Z^*)$ has density at least $2^{-2t^2}$, and it attains a higher value for the optimization problem of \eqref{equ:cube-density-optimization}. This contradicts the optimality of $C^*$. Note that \eqref{equ: f spread equ 1} is almost the condition of Definition~\ref{def:spread matrix}, with the caveat that we only consider $X'\times Y'$ of size at least $\delta$. Before we argue for rectangles $X'\times Y'$ with size smaller than $2^{-t^2}$, we first truncate $f$ and enforce an $\ell_{\infty}$ condition for it.

\paragraph*{Truncating the faces: Step 1.}
We consider the $\ell_{\infty}$-truncated version of these functions:
$$
\begin{aligned}
\overline{f}(x,y) = \min(f(x,y), 2^{t}), \\
\overline{g}(x,z) = \min(g(x,z), 2^{t}), \\
\overline{h}(y,z) = \min(h(y,z), 2^{t}).
\end{aligned}
$$
We claim that $\| \overline{f} - f \|_{L^1(X^*\times Y^*)} \le 2^{-\Omega(t^2/c)}$, and similar conditions hold for $g$ and $h$.

Indeed, consider
$$
S \coloneqq \{ (x,y): f(x,y) > 2^{t}\} \times Z^*.
$$
Note that $S$ is a slice. If $|S| > 2^{-t^2/c}|\Omega|$, it would follow that
$$
\begin{aligned}
\E_{w\sim p}[S(w)] 
&= \Pr_{w\sim p}[w\in C^*] \E_{w\sim p_{C^*}}[S(w)] \\
&= \frac{|C^*|}{|\Omega|} \E_{w\sim C^*}[p(w)] \E_{w\sim p_{C^*}}[S(w)] & \text{(By \eqref{equ:change-of-dist})} \\
&= \frac{|C^*|}{|\Omega|} \E_{w\sim C^*}[p(w)] \cdot \frac{|S|}{|C^*|} \E_{w\sim S}[p_{C*}(w)]  & \text{(By \eqref{equ:change-of-dist})} \\
&= \frac{|C^*|}{|\Omega|} \cdot \frac{|S|}{|C^*|} \E_{w\sim S}[p(w)]  \\
&> \frac{|S|}{|\Omega|} \cdot 2^t \\
&= 2^t\cdot \Pr_{w\sim \Omega}[S(w)].
\end{aligned}
$$
This contradicts our assumption on $p$. Therefore we must have $|S| < 2^{-t^2/c}|\Omega|$. 

Now we know that $S$ must be too small so that it is not subject to the evasiveness condition. Still, we can include some extra pairs $(x,y)$ into $S$ to get a larger slice function $\overline{S'} = S_{xy}\times Z^*$ such that $|S'|\approx 2^{-t^2/c}|\Omega|$. Noting that $f(x,y) = \overline{f}(x,y)$ for every $(x,y)\notin S_{xy}$, we proceed as
$$
\begin{aligned}
&~~~~ \| \overline{f} - f \|_{L^1(X^*\times Y^*)} \\
&= \E_{(x,y)\sim X^*\times Y^*}[|f(x,y) - \overline{f}(x,y)|] \\
&= \frac{|S_{xy}|}{|X^*\times Y^*|} \E_{(x,y)\sim S_{xy}}[|f(x,y) -\overline{f}(x,y)|] \\
&\le \frac{|S_{xy}|}{|X^*\times Y^*|} \E_{(x,y)\sim S_{xy}}[f(x,y)] \\
&= \frac{|S_{xy}|}{|X^*\times Y^*|} \E_{(x,y,z)\sim S}[p_{C^*}(x,y,z)]  \\
&\le \frac{|S_{xy}|}{|X^*\times Y^*|} \E_{(x,y,z)\sim S}[p(x,y,z)] 
& \text{(Recall $p_{C^*} \le p$ pointwise)} \\
&\le 2^{t-t^2/c}. 
& \text{(By evasiveness)}
\end{aligned}
$$
\paragraph*{Truncating the faces: Step 2.} To apply Theorem~\ref{thm: product of spread matrices}, we also need to enforce the left-min-degree condition on $\overline{f}$ and $\overline{h}$. Toward that goal, we simply remove all $x\in X^*$ such that $\E_{y}[\overline{f}(x,y)] < (1-\eps)\|\overline{f}\|_1$ and similarly all $z\in Z^*$ such that $\E_y[\overline{h}(z,y)] < (1-\eps) \|\overline{h}\|_1$.

We actually only remove very few $x\in X^*$ and $z\in Z^*$ in this process. To see this, we note that the original $f$ satisfies a stronger min-degree assumption: namely we have $\E_y[f(x,y)] \ge (1-\eps/2)\E[f]$. Indeed, suppose this fails at a row $x$. Let us consider $X' = X^*\setminus \{x\}$ and calculate the average of $f$ on $X'\times Y$:
\begin{align*}
    &~~~~ \E_{(x,y)\sim X'\times Y}[f(x,y)]  \\
    &> \frac{|X^*|}{|X^*| - 1} \left( 1 - \frac{(1-\eps/2)}{|X^*|}  \right) \\
    &= \frac{|X^*|-1+\eps/2}{|X^*| - 1} \\
    &= 1 + \frac{\eps/2}{|X^*|-1}.
\end{align*}
Thus, it is clearly seen that
\begin{align*}
    &~~~~ \E_{(x,y)\sim X'\times Y}[f(x,y)] \cdot \left( \frac{|X'|}{|X^*|} \right)^{1/t} \\
    & \ge \left( \left( 1 + \frac{\eps/2}{|X^*|-1} \right)^{t} \cdot \left( 1- \frac{1}{|X^*|} \right) \right)^{1/t} > 1.
\end{align*}
This shows that $f$ is not spread, as witnessed by $X'\times Y$, a contradiction, contradicting to \eqref{equ: f spread equ 1}.

We have shown that $\overline{f}$ is obtained from $f$ be removing at most $2^{t-t^2/c}\E[f]$ ``mass'' (measured in $\ell_1$) from $f$. On the one hand, the decrease of $\E[\overline{f}]$, as compared with $\E[f]$, is negligible. On the other hand, one needs to remove $\frac{\eps}{2}\E[f]$ of the ``mass'' from a row $x$ to compromise the min-degree condition on that row. Thus, a simple application of Markov's inequality implies that at most $2^{-\frac{t^2}{2c}}\cdot |X^*|$ rows are removed in Step 2.

\paragraph*{Scaling back.} After the two steps, say we currently have functions $\overline{f},\overline{g}, \overline{h}$ and the ambient space $\overline{X}\times Y^*\times \overline{Z}$ ($\overline{X},\overline{Z}$ may be ever so slightly smaller than $X^*,Z^*$). We scale up $\overline{f},\overline{g},\overline{h}$ and obtain:
$$
\begin{aligned}
\tilde{f} = \overline{f} / \| \overline{f} \|_1 , ~~~ \tilde{g} = \overline{g} / \|g\|_1 ~~~, \tilde{h} = \overline{h} / \|\overline{h}\|_1.
\end{aligned}
$$
Note that $\|\overline{f}\|_1$ is understood as $\E_{(x,y)\sim \overline{X}\times Y^*}[\overline{f}(x,y)]$ (and similarly for $\|\overline{g}\|_1, \|\overline{h}\|_1$).

We now prove that the functions $\tilde{f},\tilde{g},\tilde{h}$ are almost as spread as $f,g,h$. We prove the claim for $\tilde{f}$ as the claims for $\tilde{g},\tilde{h}$ are similarly established. The only property we need here is that $\tilde{f},\tilde{g}$ and $\tilde{h}$ are extremely close to $f,g,h$ in $\ell_1$ distance. Take $\tilde{f}$ as a demonstrating example: for every $X'\times Y'$ of density $\delta > 2^{-t^2/4c}$, we have
$$
\begin{aligned}
\E_{w\sim X'\times Y'}[\tilde{f}(w)] 
&\le \| f - \tilde{f} \|_1 \cdot \frac{|X^*\times Y^*|}{|X'\times Y'|}  +  \E_{w\sim X'\times Y'}[ f(w) ] \\
&\le 2^{-\frac{t^2}{2c} + \frac{t^2}{4c} + O(t)}  +  \delta^{-1/t} \E_{w\sim X^*\times Y^*}[ f(w) ] \\
&\le \delta^{-1/2t} \E_{w\sim X^*\times Y^*}[\tilde{f}(w)].
\end{aligned}
$$
Also, for $X'\times Y'$ of density $\delta < 2^{-t^2/4c}$, we have $\|\tilde{f}|_{X'\times Y'}\|_1 \cdot \delta^{1/t} \le \|\tilde{f}\|_{\infty} \cdot \delta^{1/t} \le \|\tilde{f}\|_1$. These two cases combined imply that $\|\tilde{f}\|_{\overline{t},\overline{t}}\le (1+\eps)\|f\|_1$.

Recall that $\|\tilde{f}\|_1 = 1$ and $\|\tilde{f}\|_\infty \le 2^t$. We now make use of Lemma~\ref{lemma: grid upper bound for spread matrix} and conclude that
$$
\| \tilde{f} \|_{U(k,2)} \le (1+\varepsilon) \| \tilde{f}\|_1
$$
for some $k = \Theta(t/\eps)$.

\paragraph*{Lower-bounding the cylinder intersection.} We are finally ready to lower-bound the size of the cylinder intersection $F$. Let $\overline{C} = \overline{X}\times Y^* \times \overline{Z}$. Recall we have the bound that $\|\tilde{f}\|\le 2^{t}\cdot (1+o(1))$ (and similarly for $\tilde{g}$ and $\tilde{h}$. We thus get the following pointwise bound:
\begin{align}
F(x,y,z) 
\ge \frac{1}{4\cdot 2^{3t}} \tilde{f}(x,y)\cdot \tilde{g}(x,z)\cdot \tilde{h}(y,z) \label{equ: lower bound F to do}
\end{align}
Write $(\tilde{f}\circ \tilde{h})(x,z) = \E_{y\sim Y^*}[\tilde{f}(x,y)\cdot \tilde{h}(z,y)]$. Applying Theorem~\ref{thm: product of spread matrices}, we conclude that
\begin{align*}
    \|\tilde{f} \circ \tilde{h} - R_{\tilde{f}} \circ R_{\tilde{h}} \|_k \le O(\eps^2 \|\tilde{f}\|_1\|\tilde{h}\|_1).
\end{align*}
Therefore, we have
\begin{align*}
     &~~~~ \E_{(x,y,z)\sim \overline{C}}[\tilde{f}(x,y)\cdot \tilde{g}(x,z)\cdot \tilde{h}(y,z)] \\
    &= \left\langle \tilde{f}\circ \tilde{h} , g \right\rangle \\
    &\ge \left\langle R_{\tilde{f}} \circ R_{\tilde{h}}, g \right\rangle - \| \tilde{f} \circ \tilde{h} - R_{\tilde{f}} \circ R_{\tilde{h}} \|_k \cdot \| g \|_{k^\star} \\
    &\ge (1-O(\eps)) \|\tilde{f}\|_1\|\tilde{g}\|_1\|\tilde{h}\|_1 - O(\eps^2 \|\tilde{f}\|_1 \|\tilde{g}\|_1 \|\tilde{h}\|_1 ) \\
    &\ge  (1-O(\eps)) \|\tilde{f}\|_1\|\tilde{g}\|_1\|\tilde{h}\|_1.
\end{align*}
Plug this back in \eqref{equ: lower bound F to do}, we get
\begin{align*}
    \E_{(x,y,z)\sim \overline{C}}[F(x,y,z)] \ge \Omega\left( 2^{-3t} \right).
\end{align*}
Combining this with the bound $|\overline{C}| \ge \Omega(2^{-t^2} |\Omega|)$, the proof is complete.

\section{Number-on-Forehead Communication Lower Bound}

 We describe a couple of consequences of Theorem~\ref{thm: removal lemma}. The first is a quantitative improvement of the main pseudorandomness result from \cite{KelleyLM2024-nof}.

\begin{theorem}\label{thm: pseudorandom for CI from cubes}
There is an absolute constant $c > 0$ for which the following is true. Let $t\in \mathbb{N}$, and let $D\subseteq [N]^3$ be a set that is $(c\sqrt{t},\frac{t}{c})$-evasive w.r.t.~slice functions. For any cylinder intersection $F$ of density $\E[F(x,y,z)] \le 2^{-t}$, it holds that
\begin{align*}
\E_{(x,y,z)\sim D}[F(x,y,z)] \le 2^{-\Omega(\sqrt{t})}.
\end{align*}
\end{theorem}

\begin{proof}
    By Theorem~\ref{thm: fractional cover}, we can write $F \le \sum_{i=1}^{R} c_i \cdot s_i(x,y,z)$ where each $s_i$ is of size at least $2^{-O(t)}$ and $\sum_i c_i \E[s_i] \le 2^{-\Omega(\sqrt{t})}$. Then, by the evasiveness of $D$ and assuming $c > 0$ is small enough, we obtain:
    \begin{align*}
        \E_{w\sim D}[F(w)] 
        &\le \sum_{i=1}^{R} c_i \E_{w\sim D}[s(w)] \\
        &\le \sum_{i=1}^{R} c_i \cdot 2^{c\sqrt{t}} \cdot \E_{w\sim [N]^3}[s(w)] \\
        &\le 2^{-\Omega(\sqrt{t})},
    \end{align*}
    as claimed.
\end{proof}

We cite a construction of evasive sets from \cite{KelleyLM2024-nof}.

\begin{theorem}[See Lemma~2.16 and Lemma 2.17 of \cite{KelleyLM2024-nof}]\label{thm: construction of pseudorandom sets}
Let $q\ge 1$ be a prime power and $\mathbb{F}_q$ be the finite field of order $q$. Let $k\ge 1$ and $N = q^k$. Identify $[N]$ with $\mathbb{F}_q^k$. Then, the set
\begin{align*}
    D = \{ (x,y,z) \in [N]^3: \langle x,y\rangle = \langle x,z\rangle = \langle y, z\rangle \}
\end{align*}
is $(O(1), \Omega(k\log(q)))$-evasive against slice functions. Here $\langle x,y\rangle := \sum_{i=1}^{k}x_i y_i \in \mathbb{F}_q$ denotes the inner product between $x$ and $y$.
\end{theorem}

We sketch the proof of Theorem~\ref{thm: construction of pseudorandom sets} below and refer the readers to \cite{KelleyLM2024-nof} for the complete proof of a stronger statement. Note that here we only need to establish the evasiveness of $D$ against slice functions. In \cite{KelleyLM2024-nof}, a much stronger pseudorandomness property of $D$ is proved.

\begin{proof}[Proof sketch.]
    Let $S = S_{xy} \times S_z \subseteq (X\times Y) \times Z$ be a relatively large slice function (i.e., of size at least $q^{-ck}\cdot N^3$ where $c > 0$ is a small but absolute constant). The evasiveness condition on $S$ is equivalent to $\frac{|S\cap D|}{|D|} \le O\left( \frac{|S|}{|\Omega|} \right)$. So, to prove Theorem~\ref{thm: construction of pseudorandom sets} we must upper-bound $\frac{|S\cap D|}{|S|}$ by roughly $\frac{|D|}{|\Omega|}\approx q^{-2}$. We may decompose $S_{xy} = \bigsqcup_{c\in \mathbb{F}_q} S_{xy}^{(c)}$ where $S_{xy}^{(c)} := \{(x,y)\in S_{xy} : \langle x, y\rangle = c\}$ and write $|S\cap D| = \sum_{c} |\{(x,y,z)\in S_{x,y}^{(c)}\times S_z: \langle x,z\rangle = \langle y,z\rangle = c\}|$. For each $c\in \mathbb{F}_q$, we upper-bound the term by applying counting lemmas between $z$ and $x,y$ (Alternatively, one can appeal to the property that Inner Product is a good extractor).
\end{proof}

We prove Theorem~\ref{thm: main nof lb} as a corollary below. 

\begin{proof}[Proof of Theorem~\ref{thm: main nof lb}]
The randomized NOF communication upper bound can be obtained by running a (randomized) equality check among three players.

To see the deterministic (or non-determinisitic) NOF lower bound, it is well-known that if $D$ is solved by a non-deterministic NOF protocol with complexity $c\log^{1/2}(N)$ where $c > 0$ is sufficiently small, then there is a cylinder intersection $F\subseteq D$ that covers at least $2^{-c\log^{1/2}(N)}$ points of $D$. Namely, for $t = c^2 \log(N)$, we have
\begin{align*}
    \E_{(x,y,z)\sim D}[F(x,y,z)]  >  2^{-\sqrt{t}}.
\end{align*}
Using Theorem~\ref{thm: pseudorandom for CI from cubes} in its contrapositive form, this implies that $\frac{|F|}{|\Omega|} \ge 2^{-t} \ge N^{-O(c)}$. For $c > 0$ small and $k$ (as in Theorem~\ref{thm: construction of pseudorandom sets}) large enough, we have $\frac{|F|}{|\Omega|}  \ge N^{-O(c)} \gg N^{-2/k} \approx \frac{|D|}{|\Omega|}$. This implies that $F$ cannot be a subset of $D$, a contradiction.
\end{proof}

\section*{Acknowledgment}

We thank Anthony Ostuni and Amit Chakrabarti for valuable comments on an earlier version of the paper. We are also grateful to the anonymous FOCS reviewers for their useful feedback.

\bibliographystyle{alpha}
\bibliography{references}

\end{document}